%% file: main.tex
\documentclass[runningheads]{llncs}
\usepackage[T1]{fontenc}

\usepackage{tmpl}
\usepackage[natbib=true]{biblatex}
\usepackage{etoolbox}
\AtEndEnvironment{proof}{\hfill$\square$}

\begin{document}
\title{Price of Anarchy of Algorithmic Monoculture}
\author{Robert Kleinberg \and
Erald Sinanaj \and
Éva Tardos}
\authorrunning{R. Kleinberg, E. Sinanaj and É. Tardos}
\institute{Cornell University, Ithaca NY, USA\\
\email{rdk@cs.cornell.edu, erald@cs.cornell.edu, eva.tardos@cornell.edu}}
\maketitle              
\begin{abstract}
Several recent works investigate the effects of \emph{monoculture}, the ever increasing phenomenon of (possibly) self-interested actors in a society relying on one common source of advice for decision making, with an archetypal driving example being the growing adoption and predictive power of machine learning models in matching markets, e.g. in hiring. Kleinberg and Raghavan (PNAS, 2021) introduced a model that captures the effects of monoculture in a one-sided matching market with advice, demonstrating that a higher accuracy common signal (such as an algorithmic vendor) might incentivize society as a whole to rationally adopt it, but as a collective it would be better off if each instead adopted less accurate, but private advice. We generalize their model and address the open question of their work in quantifying the \emph{social welfare loss}. We find that monoculture and more generally decentralized optimization is close to optimal: we show a tight constant bound of 2 on the price of anarchy (and more general notions) for the induced game.
\end{abstract}

\input{1intro}
\input{2related}
\input{3prelim}

\input{4mainResults}
\input{5assumptions}

\input{6lowerBounds}
\input{7conclusion}
\subsubsection{\ackname} Erald Sinanaj was supported in part by AFOSR grant FA9550-23-1-0410. Éva Tardos was supported in part by AFOSR grants FA9550-23-1-0410 and FA9550-231-0068, and ONR MURI grant N000142412742.
\printbibliography
\input{8apdx}

\begin{credits}
\end{credits}
\end{document}

%% file: 1intro.tex
\section{Introduction}

Algorithmic monoculture describes the now prevalent pattern of collective reliance on algorithmically curated advice. 
While the predictive power of algorithms and machine learning models usually outclasses traditional means, convergence and over-dependence can lead to detrimental effects in social welfare \cite{monoculture_kleinberg_raghavan,monoculture_matching_peng_garg,improved_bayes_risk_jagadeesan} and fairness concerns over outcome homogenization \cite{homegenization_analysis_toups,outcome_hom_bommasani,stat_feedback_loop_baek}.

Our work considers a one-sided matching market under limited advice and specifically a model that is directly influenced by the one introduced by Kleinberg and Raghavan in \cite{monoculture_kleinberg_raghavan}, that kick-started the study of the effects of algorithmic monoculture. 
The motivation stemmed from the setting of \emph{algorithmic hiring}: in their model a number of firms each seek to hire one candidate out of a pool of many, sharing a common ground truth valuation. 
However, they do not know \emph{a priori} this valuation (or future performance). 
To the aid of their decisions, they have access to a public algorithmic vendor, that relays a single shared noisy ranking to any who choose to view it, and a private human advisor (employee) that also produces a ranking for their corresponding firm, independently of any other. 
They model both of them as ranking distributions of the same family, parameterized by an accuracy term; the algorithm is more accurate than a human advisor. 
Each firm chooses either the algorithm or the human for advice, and they proceed in random order to hire the top remaining candidate in the ranking they received. 

The theme of \cite{monoculture_kleinberg_raghavan} is, as they state, a counterintuitive fact: the choice of the algorithm might be a dominant strategy for each firm, but welfare and thus average utility might be better if every firm instead used the less accurate human evaluator. 
They show this to be the case for two firms, under some natural assumptions, and experimentally verify that this phenomenon can arise for multiple firms as well. 
However, this surprising qualitative conclusion does not shed light to the degree of the negative effect that monoculture impacts. 
Our work is devoted to providing a measure to this welfare loss, resolving their open question with a rather reassuring answer: monoculture (and more generally selfish behavior) is not that bad. 
We show that the \emph{price of anarchy} for this game is 2 (which is asymptotically tight), for \emph{any} number of agents and \emph{any} number of ``advice sources'', thus bounding the welfare loss even when there is no dominant strategy (and we also extend this for more general notions of equilibria). 

We show this to be the case for a generalization of their model, substantially relaxing the assumptions about the structure of the advice acquisition and the relevant ranking distributions, which allows us to apply our findings to a much broader landscape, capturing a wealth of realistic scenarios. 
The properties we seek for the advice sources are attained by a very rich class of distributions, which we term \emph{Schur-concave} rankings; as an example, any ranking produced by additive i.i.d. noise falls under this class, whenever the noise distribution is \emph{log-concave} (with two examples being the Gaussian and Laplacian, considered in \cite{monoculture_kleinberg_raghavan}). 
We also show that the quite well known Mallow's Model rankings enjoy the same properties that lead to our bounds. 
As previously mentioned, our results do not assume the existence of a dominant strategy and allow for arbitrary asymmetry between the firms, specifically any number and combination of unique idiosyncratic and `common' signals. 
Furthermore, we extend our results to a more realistic mechanism, that allows for more expressive strategic behavior and we show that under minimal assumptions a firm does not have to \emph{strategize} any differently from the previous model. In addition, our main results assume a condition which roughly states that for any two candidates it is more likely that their relative order is correct, rather than not. For environments that are \emph{systemically} biased, in that advice sources may violate this property, we provide price of anarchy bounds that scale gracefully with the degree of violation of this condition.

The paper is structured as follows. 
In the immediately succeeding section we discuss related work. In \Cref{sec:3prelim} we formalize specifics of our problem. 
In \Cref{sec:4poa} we show our main result; a constant price of anarchy. 
In \Cref{sec:5assumptions} we discuss the generality of our assumptions and the strategic extension. 
In \Cref{sec:lowerBounds} we establish that our price of anarchy bound is tight and show that if the natural assumption we impose on the rankings is arbitrarily violated, the price of anarchy becomes unbounded. 
Finally, in \Cref{sec:relaxation-of-sc} we introduce the relaxation of 
our condition on rankings, and go on to provide a price of anarchy bound which degrades in a controlled manner, as a function 
of the extent to which the original assumption is violated.

%% file: 2related.tex
\section{Related Work}
\label{sec:2related}
A number of recent papers explore the theme of \emph{algorithmic monoculture}, addressing various questions about the impacts of substituting algorithmically-derived rankings and recommendations for human preferences in strategic settings. 
These works typically model the algorithmically-derived  preferences as more accurate (which tends to improve social welfare) but less varied (which tends to reduce social welfare), leading to subtle questions about which of these effects predominates. 
This phenomenon was explored in the setting of one-sided matching by \citet{monoculture_kleinberg_raghavan}, in two-sided matching markets by \citet{monoculture_matching_peng_garg}, and in the setting of competing classifiers by \citet{improved_bayes_risk_jagadeesan}. 
While each of the aforementioned papers introduces examples of settings in which algorithmic monoculture leads to reduced social welfare, to the best of our knowledge the present work is the first to bound the worst-case social welfare loss due to algorithmic monoculture. 
We establish such bounds in a one-sided matching model directly inspired by~\cite{monoculture_kleinberg_raghavan}.

The model of one-sided matching we investigate assumes that the mechanism used is Random Serial Dictatorship (also known as Random Priority). While the price of anarchy of one-sided matching mechanisms, including Random Serial Dictatorship, has been investigated in several prior works, the model we study differs in two key respects. 
First, it is a common value environment: any candidate has the same value to all agents. Consequently, the social welfare of a matching depends only on the set of candidates matched, rendering price-of-anarchy questions vacuous in the special case when  candidates and agents are equinumerous. 
Second, most importantly, the agents in our model are uncertain of the candidates' values. 
In fact, the main strategic decision they face is not which candidate to select, but which \emph{signal} to obtain in order to advise their ranking of the candidates.

In the classical model of matching when firms have private values for the candidates, for markets with an equal number of agents and candidates, under a standard unit-sum normalization assumption on agents' values, \citet{welfare_onesided_christodoulou} obtained a tight bound of $\Theta(\sqrt{n})$ on the price of anarchy of both Random Serial Dictatorship and another widely-studied mechanism, the Probabilistic Serial mechanism of \citet{probabilistic_serial_bogomolnaia}. 
Other earlier works had focused on welfare guarantees for truthful mechanisms \citep{efficiency_adamczyk,welfare_onesided_RSD_filosRatsikas} and some investigated novel measures of welfare \citep{welfare_onesided_without_money_bhalgat}. \citet{poa_asymmetric_probabilistic_serial} analyzed the price of anarchy of one-sided matching mechanisms in asymmetric markets where candidates may outnumber agents. 
In particular, they showed that the price of anarchy in markets with $n$ agents can be asymptotically greater than $\sqrt{n}$ when the number of candidates is sufficiently large compared to $n$.

%% file: 3prelim.tex
\section{{Setup}}
\label{sec:3prelim}
{\bf Game specification}
There are $n$ firms (agents) with a shared but unknown cardinal preference over $m$ candidates (items) and a unit demand over them. Any candidate awards a firm some non-negative utility or value. 
Firms have access to \emph{ranking technologies} (advice sources). 
Each firm chooses one such technology, from those available to her and receives a ranking sample, not necessarily independent of one another firm receives (below we discuss this in further length). 
\emph{Random Serial Dictatorship} dictates the order in which each firm is called: nature draws a uniformly random ordering and each firm is called in that order. When a firm is called, she receives the top remaining candidate from her chosen ranking and this is public knowledge. 
This is the setting of \cite{monoculture_kleinberg_raghavan}, which we will mostly focus on; we will call this mechanism \emph{Obedience-constrained} Random Serial Dictatorship. 
We will also investigate a less strict mechanism, which is identical to the previous but different in one important aspect: a firm can decide freely about her choice of a candidate when it's her turn. 
We will call this mechanism \emph{Unconstrained} Random Serial Dictatorship. 
More discussion about the unconstrained model will follow in \Cref{sec:5assumptions}; unless explicitly stated, all definitions and claims will be for the former model.

{\bf Ranking technologies}
A ranking technology is a distribution on rankings of candidates (permutations of the set $[m] := \{1, \hdots, m\}$). 
There exists a \emph{common signal space} $\cal{A}$ of independent ranking technologies, such that if any two firms choose the same strategy from $\cal{A}$ they receive the same sample. 
Each firm has access to some arbitrary subset $\cal{A}_i \subset \cal{A}$ and an \emph{idiosyncratic} space of independent ranking technologies $\cal{H}_i$, which is mutually disjoint between firms, (that is $\cal{H}_i \cap \cal{H}_j = \emptyset$ for $i \neq j$). 
We note that ranking technologies are mutually independent (for any realization of the candidate values). 
The action space of the firms in line with the model of \citet{monoculture_kleinberg_raghavan} corresponds to $\cal{A}$ and $\cal{H}_i$ being singleton sets, with $\cal{A}_i = \cal{A}$ and each idiosyncratic technology being identically (but independently) distributed.
The pure strategy space is the set $\cal{S} := \times_{i \in [n]} \cal{R}_i$ (with $\cal{R}_i := \cal{A}_i\cup\cal{H}_i$). Further, we will refer to $\cal{R} := \bigcup_{i \in [n]}\cal{R}_i$ as the \emph{advice space} of the game.

{\bf Valuations, welfare and notation}
Candidates are identified with indices $i\in [m]$ and their value $x(i)$. 
The joint value vector $\mbf{x}$ is drawn from some distribution $\cal{X}$ supported on a subset of $[0, \infty)^m$.
Whenever we have a ranking technology $R$ we will usually denote with $\mbf{r} = (r_1, \hdots, r_m)$ the ranking (i.e., permutation of $[m]$) sampled from the corresponding distribution $\cal{F}_R(\mbf{x})$. 
If $C$ is some subset of $[m]$, $\mbf{r}^{-C}$ will denote the partial ranking induced by removing the candidates in $C$ from consideration. 
Thus, in line with the previously defined notation, $r^{-C}_i$ denotes the $i$-th ranked candidate in the corresponding ranking (whenever it's well defined). 
We will usually denote a pure profile by the letter $\mbf{s} \in \cal{S}$. 
The utility $u_i(\mbf{s})$ of some firm in profile $\mbf{s}$ is the expected value of the candidate she hires and the social welfare $SW(\mbf{s})$ is the usual \emph{utilitarian} one (the utility sum). 
A socially optimal profile is one that maximizes the social welfare. 
The Price of Anarchy \cite{poa_koutsoupias} is the ratio of the optimal welfare to the worst-case Nash equilibrium welfare.

%% file: 4mainResults.tex
\section{The price of anarchy}
\label{sec:4poa}
Before we state our main result we define a very natural property that we would like a ranking technology to have.

\begin{definition}[Stochastically Consistent]\label{def:sc}
    A ranking technology $R$ is \emph{stochastically consistent} iff the following holds for all indices $i < j \in [m]$ and candidates $k, \ell \in [m]$. The fact $x(k) \ge x(\ell)$ implies that:
    \[\pr[r\sim \cal{F}_R(x)]{(r_i, r_j) = (k, \ell)}[r^{-\set{k, \ell}}]\\
    \geq\\
    \pr[r\sim\cal{F}_R(x)]{(r_i, r_j) = (\ell, k)}[r^{-\set{k, \ell}}] \text{ a.s.}\]
\end{definition}
That is, the ranking obeys a notion of \emph{independence of irrelevant alternatives}: conditioned on two candidates $k$ and $\ell$ being ranked at two specific spots ($i$ and $j$) 
in any order, and conditioned on the set of candidates ranked in some other arbitrary positions, it is more probable that the ranking correctly represents the relative order of the aforementioned pair, rather than not. 
In \Cref{sec:5assumptions} we show how deceivingly broad it is; demonstrating that it captures almost any natural ranking distribution one can think of and the wealth of those studied in the literature.

We extend this definition, declaring an \emph{advice space} stochastically consistent when all its ranking technologies are stochastically consistent. 
Our main result is the following favorable conclusion about the social welfare loss of monoculture and more generally of any Nash equilibrium.
\begin{theorem}\label{thm:PoA}
   The Price of Anarchy of the Obedience-Constrained Random Serial Dictatorship mechanism, in which firms obtain their rankings from a stochastically consistent advice space, is bounded above by 2.
\end{theorem}

Before presenting the full proof of the theorem, we sketch the intuition behind the result. 
An initial observation is that, since agents have shared (cardinal) preferences in our model, the total utility of all agents is equal to the total value of the candidates
hired. 
The main idea of the proof is to consider two cases, depending on the total value 
of the candidates that each firm gets in the socially optimal profile $\mbf{s}^*$, but would have been ``taken'' from them in the equilibrium profile $\mbf{s}$ by the firms that went before her. The value of this set of candidates will be denoted in the proof by $W^{(\beta)}_{\text{snatched}}(\mbf{s}^*, \mbf{s})$, where $\beta$ is the order of the firms. 
If its expectation 
is at least half that of the optimum social welfare, then clearly the price of anarchy is bounded above by two. 
Otherwise, at least half of the expected total value of the candidates hired in the socially optimal profile must be attributable to candidates that would not be ``taken'' by those before the respective firm in the equilibrium profile.


Thus, in the second case, if a firm at random in the equilibrium had unilaterally deviated to her strategy in the optimal profile, she would be expected to do reasonably well by at least gaining a fraction of this ``unshared'' welfare, assuming that rankings are ``well behaved'' enough (which we discuss immediately below). 
The argument concludes by recalling the definition of the equilibrium; the social welfare of the equilibrium should upper bound the total expected utility of every such firm unilaterally deviating to this strategy. 

The next lemma establishes the ``well behaved'' property of the rankings that we need for our argument to hold and shows that it is the case for our modeling assumptions. 
In short, we claim that \emph{stochastic consistency} of the ranking strategies implies the following fact. For any two pure strategy profiles $\mbf{s}$, $\mbf{s}^*$, conditioned on the order of firms being $\beta$, whenever the random candidate that firm $i$ hires in profile $\mbf{s}^*$ is not selected by the firms before her in $\mbf{s}$, she would have done as well or better (than her result in the optimal profile $\mbf{s}^*$), competing in profile $\mbf{s}$, if she was using $\mbf{s}_i^*$. 
While this idea is intuitive, it turns out to not be the case in general: in some cases a higher ranked candidate being available may correlate with the value of the candidate being poor in expectation. 
We show that our intuition is true however, when the advice space is stochastically consistent. The proof of the statement below can be found in the Appendix. 

\begin{restatable}{lemma}{lemmaDev}\label{lemma:deviation}
Consider an Obedience-Constrained Random Serial Dictatorship mechanism with a stochastically consistent advice space. 
Let $c_i^{(\beta)}(\mbf{s})$ denote the candidate hired  by firm $i$ in profile $\mbf{s}$ when the order of firms is $\beta$. 
Let also $x_i^{(\beta)}(\mbf{s}) := x(c_i^{(\beta)}(\mbf{s}))$ be the value of this candidate. 
Further, define $C_i^{(\beta)}(\mbf{s})$ to be the set of candidates hired by the firms going before $i$ in profile $s$.\\

For any two pure strategy profiles $s$, $s^*$ the below holds:
\[\ex{x_i^{(\beta)}(s_i^*, \mbf{s}_{-i}) - x_i^{(\beta)}(\mbf{s}^*)}[\mbf{x}, \beta, c_i^{(\beta)}(\mbf{s}^*) 
\notin C_i^{(\beta)}(\mbf{s})] \geq 0\]
where the expectation is over the randomness of the samples drawn from the ranking technologies and $\mbf{x}$ represents the joint value vector of the candidates.
\end{restatable}



\newcommand{\snatched}{W^{(\beta)}_{\mathrm{snatched}}(\mbf{s}^*, \mbf{s})}
\newcommand{\avail}{W^{(\beta)}_{\mathrm{available}}(\mbf{s}^*, \mbf{s})}

Now we are ready to prove our main result.
\begin{proof}[\Cref{thm:PoA}]
Consider a pure strategy profile $s=(s_1, \hdots, s_n)$. 
The outcome of the selected candidates depends on the rankings sampled from the technologies adopted by the agents, but also on the order $\beta$ that Random Serial Dictatorship draws (and the random candidate values $\mbf{x}$). 
Let the socially optimal profile be $s^*$ and define the following quantity the we previously described informally:
\[\snatched := \ex{\sum_{i \in [n]}x_i^{(\beta)}(\mbf{s}^*)\cdot\ind{c_i^{(\beta)}(\mbf{s}^*) \in C_i^{(\beta)}(\mbf{s})}}[\text{RSD order } = \beta]\]
(Note that when $\beta(i) = 1$ we have $C_i^{(\beta)}(\mbf{s}) = \emptyset$ and thus these terms are zero.)
Similarly, define: 
\[\avail := \ex{\sum_{i \in [n]}x_i^{(\beta)}(\mbf{s}^*)\cdot\ind{c_i^{(\beta)}(\mbf{s}^*) \notin C_i^{(\beta)}(\mbf{s})}}[\text{RSD order } = \beta]\]

To illustrate this, consider the example in \Cref{fig}. Suppose the order of the firms is fixed, and in that order from top to bottom we have, on the left, the candidates hired in profile $\mbf{s}$, and on the right, those hired in profile $\mbf{s}^*$. We have colored the candidates hired in $\mbf{s}^*$ either green or red. A green-colored candidate is one that would be still available to the corresponding firm in profile $\mbf{s}$, while a red-colored candidate is one would have been “snatched" from the corresponding firm, in profile $\mbf{s}$ and not available to be hired.

\begin{figure}[htp]
    \centering
    \includegraphics[width=0.5\linewidth]{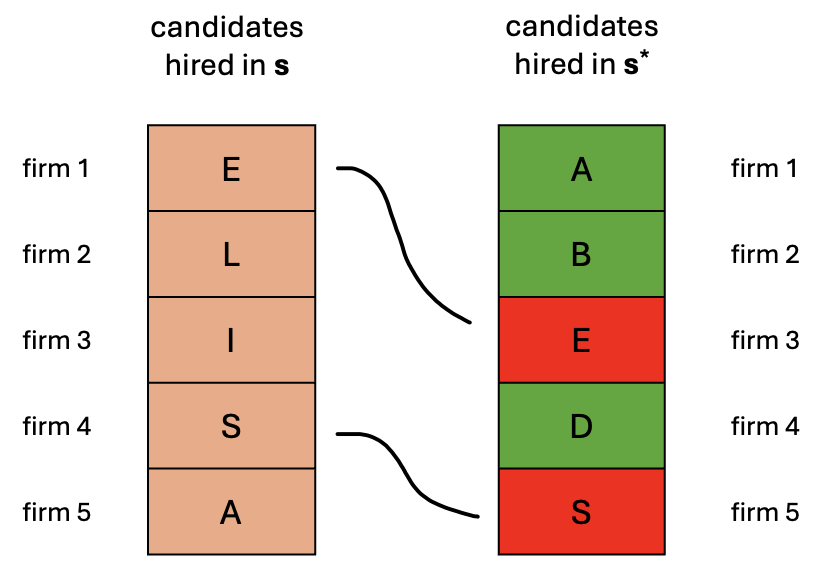}
    \caption{Candidates E and S are \textcolor{red}{snatched}, while A, B, D are \textcolor{ForestGreen}{available} to the corresponding firms.}
    \label{fig}
\end{figure}

Another way to define $W^{(\beta)}_{\text{snatched}}(\mbf{s}^*, \mbf{s})$ and $W^{(\beta)}_{\text{available}}(\mbf{s}^*, \mbf{s})$ is the following. Consider the two (random) random sequences $h^{(\beta)}_\mbf{s}, h^{(\beta)}_{\mbf{s}^*}$ of candidates hired in profile $\mbf{s}$ and $\mbf{s}^*$, according to the order of the firms being $\beta$. The quantity $W^{(\beta)}_{\text{snatched}}(\mbf{s}^*, \mbf{s})$ represents the expected welfare of the candidates of $h^{(\beta)}_{\mbf{s}^*}$ that appear earlier in $h^{(\beta)}_\mbf{s}$, while $W^{(\beta)}_{\text{available}}(\mbf{s}^*, \mbf{s})$ is the value of the rest of the candidates in $h^{(\beta)}_\mbf{s}$. Now, we distinguish two cases.

\paragraph{Case 1:} When \underline{$\ex[\beta]{\snatched} \geq \sfrac{1}{2}\cdot SW(\mbf{s}^*)$}:\\

In this case, we have $SW(\mbf{s}) \ge \ex[\beta]{\snatched} \geq \sfrac{1}{2}\cdot SW(\mbf{s}^*)$, as claimed. 
The first inequality holds because $SW(\mbf{s})$ represents the expected combined value of \emph{all} candidates
hired in profile $\mbf{s}$, whereas $\ex[\beta]{\snatched}$ represents the expected combined value of
a subset of the candidates hired in profile $\mbf{s}$.

\paragraph{Case 2:} When \underline{$\ex[\beta]{\snatched} < \sfrac{1}{2}\cdot SW(\mbf{s}^*)$}:\\

In this case, we have that equivalently: $\ex[\beta]{\avail} > \sfrac{1}{2}\cdot SW(\mbf{s}^*)$.  
Now, we use a proof outline reminiscent of the classical smoothness \cite{intrinsic_robustness} framework for price of anarchy. 
Suppose a firm in the equilibrium unilaterally deviates to the strategy she employs in the optimal profile $\mbf{s}^*$. 
The idea is to observe that if her choice in the optimal profile would not be taken by those before her in the equilibrium, she would hire either the candidate she would get in the optimal profile, or someone higher up in her list. 
In the latter case, \cref{lemma:deviation} shows that stochastic consistency (of the advice space) guarantees that her expected utility can only increase that way.

More formally, by definition:
\[\sum_{i \in [n]}u_i(s_i^*, \mbf{s}_{-i}) = \ex{\sum_{i \in [n]}x_i^{(\beta)}(s^*_i, \mbf{s}_{-i})}\]
Defining  $\cal{E}_i := \left\{c_i^{(\beta)}(\mbf{s}^*) 
\notin C_i^{(\beta)}(\mbf{s})\right\}$ for brevity and recalling that valuations are non-negative, we have:
\begin{align*}
\sum_{i \in [n]}u_i(s_i^*, \mbf{s}_{-i}) &\geq \ex{\sum_{i \in [n]}x_i^{(\beta)}(s^*_i, \mbf{s}_{-i})\cdot\ind{\cal{E}_i}}\\ 
&= \E_{\mbf{x}, \beta}\sum_{i \in [n]}\pr{\cal{E}_i \mid \mbf{x}, \beta}\ex{x_i^{(\beta)}(s^*_i, \mbf{s}_{-i})}[\mbf{x}, \beta,  \cal{E}_i]\\   \intertext{And now notice that by \Cref{lemma:deviation} this is: }
&\geq \E_{\mbf{x}, \beta}\sum_{i \in [n]}\pr{\cal{E}_i \mid \mbf{x}, \beta}\ex{x_i^{(\beta)}(\mbf{s}^*)}[\mbf{x}, \beta, \cal{E}_i]\\
&= \ex[\beta]{\avail}\\
&> \sfrac{1}{2}\cdot SW(\mbf{s}^*)
\end{align*}

Now, assume $\mbf{s}$ is a pure equilibrium; then, either case 1 holds, and we immediately conclude that $SW(\mbf{s}) \geq \sfrac{1}{2}\cdot SW(\mbf{s}^*)$ or not, in which case observe that by the equilibrium property of $\mbf{s}$ it must be:
\[SW(\mbf{s})\ge \sum_{i \in [n]}u_i(s_i^*, \mbf{s}_{-i}) \underset{\text{Case 2}}{\geq} \sfrac{1}{2}\cdot SW(\mbf{s}^*)\]

Finally, notice that nothing forbids us to extend this to mixed equilibria. We just need to consider each case in expectation over the randomness of the mixed strategy equilibrium, meaning, now we define $W^{(\beta)}(\sigma\cap\mbf{s^*}) := \ex[\mbf{s} \sim \sigma]{W^{(\beta)}(\mbf{s}\cap\mbf{s}^*)}$ for a mixed strategy profile $\sigma$. The proof then follows the same recipe, distinguishing cases according to what proportion of the optimal welfare this quantity represents. 
\end{proof}
In fact, we can extend the previous proof in a similar fashion to other general equilibrium concepts (e.g. correlated equilibria) or no regret sequences. One then might wonder if our proof can be recast in the framework of “smoothness". It turns out that the answer is yes: Stochastic Consistency implies that our game is “(1,1)-smooth", as defined by \citet{intrinsic_robustness}. The reader is referred to the appendix for a \hyperref[thm:smooth]{proof} of the fact.


%% file: 5assumptions.tex
\section{Generality of our assumptions}\label{sec:5assumptions}
In this section we discuss and establish the generality of our assumptions.

First of all, we demonstrate that the most natural and widely considered ranking distributions do indeed satisfy stochastic consistency. 
We show that it holds for the well-known \emph{Mallow's Model} distribution (with very mild assumptions) and for a very rich class of ranking technologies that we term \emph{Schur-concave rankings}. In \Cref{remark} below, we explain how the specific ranking distributions considered by \citet{monoculture_kleinberg_raghavan} fit within the two families outlined here, as special cases.

\subsection{Mallow's model}
Mallow's model \cite{mallows} has a long history in the literature of decision theory. 
It dictates that the probability of observing a certain ranking $\pi$ is proportional to $\phi^{d(\pi, \pi^0)}$, where $d$ is a distance measure on permutations \footnote{A distance function on the space of permutations $\mbf{S}_n$ of $[n]$ is a non-negative and symmetric function $d : \mbf{S}_n\times\mbf{S}_n \to [0,\infty)$.}, $\pi^0$ is the ground-truth ranking and $\phi \in (0, 1]$ is a \emph{dispersion} parameter. 
One of the most frequently used distance notions is the Kendall-Tau distance \cite{kendall}, which measures the number of inversions (disagreements between ranked pairs) between two permutations.

\begin{proposition}\label{prop:Mallow}
   The Mallow's Model of rankings with the Kendall-Tau distance is \emph{Stochastically Consistent}.
\end{proposition}
\begin{proof}
 That Mallow's Model is \emph{stochastically consistent} is fairly immediate: any realization where some pair is ranked correctly can be mapped to one where they switch positions (and obviously vice-versa). 
 This introduces \emph{at least} one inversion, the one between the pair (and
 of course, does not remove any pre-existing inversions). Hence, the Kandall-Tau distance is strictly greater. 
 Therefore, witnessing this realization instead, is at most $\phi \le 1$ times as likely, than the one where the relative ordering of the pair was correct. 
 A standard argument using the law of total probability then implies the first claim.
\end{proof}

However, the claim of \Cref{prop:Mallow} extends to an even larger umbrella of ranking families. We can actually characterize completely the class of Mallow's model ranking distributions that are stochastically consistent.

\begin{theorem}
    A Mallow's model ranking distribution is Stochastically Consistent if and only if its distance function is \emph{inversion increasing}. 
    
    A distance function $d$ is \emph{increasing with respect to inversions}, if for any permutation $\pi$ and any other permutation $\pi'$, which is a result of inverting a pair of $\pi$ that is ranked correctly (relative to each other), it holds that $d(\pi, \pi^0) \le d(\pi', \pi^0)$, where $\pi^0$ is the ground truth ranking. 
\end{theorem}
\begin{proof}
    The “if" direction can be proved by tracing the proof of \Cref{prop:Mallow}. We show the “only if" direction. Suppose we have a stochastically consistent ranking distribution that follows Mallow's model with dispersion parameter $\phi \in (0,1]$ and some distance function $d : \mbf{S}_m\times\mbf{S}_m \to [0,\infty)$. Stochastic consistency states that conditional on the candidate value vector $\mbf{x}$ (which implies a ground truth ranking $\pi^0$) the following holds for any two candidates $k, \ell$.
    If $i < j \text{ and } x(k) \ge x(\ell)$, then
    \[\pr{r_i = k, r_j = \ell}[r^{-\set{k, \ell}}] \ge \pr{r_i = \ell, r_j = k}[r^{-\set{k, \ell}}] \]

    Thus, by Bayes, for any permutation \footnote{Technically speaking, for all but a measure 0 set of them. But, this has to be the empty set, as by definition, Mallow's model has full support on the set of permutations, which is finite.} $\pi$ (with $\pi_i = k, \pi_j = \ell$), we have that $i < j$ and $x(k) \ge x(\ell)$ implies that:
    \[\pr{r = \pi} \ge \pr{r = \pi'}, \text{ where $\pi'$ has $i,j$ swapped} \]

    But, since $\pr{r = \pi} = \phi^{d(\pi, \pi^0)}$ and $\phi \in (0,1]$, we must have $d(\pi, \pi^0) \ge d(\pi', \pi^0)$.
\end{proof}

 Other examples besides the Kendall-Tau distance discussed just previously, include the Spearman-Rho and Spearman-Footrule (see e.g. \cite{spearman_diaconis}), the former of which measures the squared $L_2$-distance of a ranking from the ground truth (i.e. if $\pi$ is the permutation then $d_{SR}(\pi) = \sum_{i}\l(\pi_i-\pi^0_i\r)^2$), and the latter the $L_1$-distance ($d_{SF}(\pi) = \sum_{i}\l|\pi_i-\pi^0_i\r|$). Both of these owe their name to Spearman's work on correlation coefficients \cite{spearman} and are also ubiquitously used, along with the Kendall-Tau distance. 
Distance functions that do not fall in this category usually arise from settings of combinatorial interest and not that of decision theory, such as the Cayley distance, which counts the minimum number of transpositions to make two permutations the same, or the Hamming distance, which measures how many positions in the ranking are incorrect.

The following statement outlines a very general sufficient condition for the distance function, which also substantiates the previous claims.
\begin{restatable}{proposition}{propdist}
    Let $g:[0,\infty) \to [0,\infty)$ be a non-decreasing, convex function. The distance function $(\pi,\pi') \mapsto \sum_{i=1}^mg(|\pi_i-\pi'_i|)$ is \emph{inversion monotone}.
\end{restatable}
The proof can be found in the \hyperref[prop:cvx_sum]{appendix}.

\subsection{Schur-concave rankings}
Before we define Schur-concave rankings we need the notion of an \emph{Additive noise ranking}, which ranks candidates by noisy score estimation. 
In such a ranking distribution, first, a joint noise vector $\bm{\eps}$ is sampled and then, candidates are ranked in decreasing order of their scores: $x_i + \eps_i$. 
The below definitions serve for defining \emph{Schur-concave rankings}:

\begin{definition}[Majorization]
    A vector $\mbf{x}\in \R^n$ majorizes another vector $\mbf{y}\in \R^m$ (and we symbolize $\mbf{x} \succ \mbf{y}$) whenever we have that:
    \[\sum_{i = 1}^mx_i = \sum_{i = 1}^my_i \text{, and for all } k \in [m]: \sum_{i = 1}^kx^{\downarrow}_i \geq \sum_{i = 1}^ky^{\downarrow}_i \]
    where above  $x^{\downarrow}_i$ is the $i$-th largest entry of $\mbf{x}$.
\end{definition}

\begin{definition}[Schur-Concavity]
    A function $f: \R^m \to \R$ is said to be \emph{Schur-Concave} whenever for all pairs $\mbf{x}, \mbf{y} \in \R^m$ such that $\mbf{x} \succ \mbf{y}$:
    \[f(\mbf{x}) \leq f(\mbf{y})\]
\end{definition}

A vector that is majorized by another is more "spread out" than the latter and so a first intuition for Schur-concave distributions is that they assign less probability mass to more `extreme' points.
With this information we are now ready to state our definition.
\begin{definition}[Schur-Concave Ranking]
    We call a ranking technology \emph{Schur-Concave} if the distribution of the ranking can be modeled by an additive noise ranking, such that the joint density of the noise vector is Schur-Concave.
\end{definition}

The below claim describes a quite well-known fact in majorization theory (in the language of our setting).\footnote{See e.g. \cite{schurConcave} and observe that it can be obtained as a corollary of \emph{Fact} 13.24 and \emph{Fact} 13.29.}

\begin{restatable}{proposition}{propLog}
    Additive noise rankings with identically and independently distributed log-concave noise (in each candidate) form a \emph{Schur-concave} ranking.
\end{restatable}

Some examples of log-concave distributions are possibly the most well known: the Gaussian, the Laplacian, the uniform, beta, logistic and many more. 
The below theorem establishes the breadth of the class of distributions that our notion of ``well-behavedness'' (stochastic consistency) encompasses. 
We state its proof in the Appendix.

\begin{restatable}{theorem}{thmSchur}\label{thm:Schu-concave}
    \emph{Schur-concave rankings}  are \emph{Stochastically Consistent}.
\end{restatable}

\begin{remark}\label{remark}
    \citet{monoculture_kleinberg_raghavan} proved the following. \footnote{Informal restatement of their Theorem 1, adapted to our terminology.}
    
    \emph{Consider any two-firm game with the advice space consisting of a private technology (human evaluator) for each firm (with identical distribution) and one common ranking technology (algorithm). Suppose also that the private evaluators and the algorithm produce rankings from the same (parameterized) family of distributions, where the parameter can be thought of as an accuracy. Under sufficient conditions, it can be the case that the more accurate algorithm is a dominant strategy, but the social welfare is higher when both firms choose their evaluator.}
    
    In Theorem 3, they show that when this family corresponds to the Mallow's model ranking distributions with the Kendall-tau distance, then for any accuracy of the evaluator (dispersion parameter $\phi_H$), there exists an algorithm with higher accuracy (lower dispersion parameter $\phi_A$) such that the described phenomenon is true, regardless of the candidates' distribution. In Theorem 2, they establish an analogous result for a setting in which there are three candidates, and the family is Additive noise rankings with iid Gaussian (or Laplacian) noise, and where accuracy would correspond to the inverse of the standard deviation.
    
    These ranking distributions they analyzed are special cases of the two families introduced in this section. 
\end{remark}

\subsection{Trusting a ranking blindly}
In the rest of this section we investigate the constraint that the mechanism imposes on the firms (and its relaxation). 
Specifically, it is assumed that when it's her turn, each firm follows her ranking sample, in choosing the top ranked available candidate, no matter what she knows about the state of the world. 
However, this is not generally true for rational actors, if we (naturally) assume that they have the choice of selecting from the remaining candidates however they wish. 
This persists even if they don't have a Bayesian belief about the rest of the firms (for example they might have known that all firms have access only to well-behaved algorithmic vendors), depending on the candidate distribution and thus one might question the assumption that a firm will trust her ranking ``blindly''.


This behavior is also a problem in that it complicates a firm's decision making: it would be simpler if one was reassured that the best she could do in expectation, at any point, is to follow the ranking she sampled. 
The statement below shows that a very simple assumption about the world makes this desideratum (which we will call \emph{incentive compatibility}) true. 
Under this assumption, even if a firm knew exactly what technology the ones before her used, not just what candidate they received, she is better off choosing her top ranked choice that is available in the sample she received. Below we define the assumption in question.

\begin{definition}[Permutation Invariant Candidate Distribution] A candidate distribution $\cal{X}$ is \emph{permutation invariant} iff for any permutation $\pi$ of the set $[m]$ and any candidate value vector $\mbf{x} \in [0, \infty)^m$ the following holds for the probability density function concerning the joint values:
\[f_{\cal{X}}(\mbf{x}) = f_{\cal{X}}(\pi \circ \mbf{x})\]    
where $\pi \circ \mbf{x}$ denotes the permutation of the components of $\mbf{x}$ according to $\pi$.
\end{definition}
One can readily verify that this is satisfied by any product distribution with identical univariates (case of i.i.d. values). 
More generally, the condition does not require mutual independence of the candidate values, merely a notion of unidentifiability or symmetry in the \emph{likelihood} of the realizations of the world: ie, a name, tag, or other irrelevant attribute should not predict one's potential. 


We claim the following, which asserts that under this assumption, firms in the Unconstrained Random Serial Dictatorship mechanism have no incentive to behave differently than they would have in the Obedience-Constrained Random Serial Dictatorship. 
The proof can be found in the Appendix:
\begin{restatable}{proposition}{propIC}\label{prop:ic}
Unconstrained Random Serial Dictatorship, equipped with a \emph{stochastically consistent} advice space is \emph{incentive compatible} if the candidate distribution is \emph{permutation invariant}, even if a firm knows which technologies each firm before her selected.
\end{restatable}
The main argument is that even conditioned on some specific subset of candidates being available, the firm cannot ``distinguish'' any one of them as having lesser expected value from any other and thus will not gain in expectation by strategically permuting a sample of a stochastically consistent ranking (even adaptively to which candidate is recommended by her chosen ranking).

\Cref{prop:ic} implies that the same price of anarchy bound holds for the \emph{Unconstrained} Random Serial Dictatorship, assuming that firms follow the (weakly) dominant strategy of trusting their ranking. 
Moreover, under some additional mild assumptions about the candidate distribution (assuming the distribution is \emph{atomless}) and the advice space (assuming \emph{strict} consistency: a strict inequality for candidates with different values) trusting one's ranking is \emph{strictly} dominant.


%% file: 6lowerBounds.tex
\section{Lower bounds and violation of assumptions}
\label{sec:lowerBounds}
In this section we show that the bound on the price of anarchy we presented under the considered model assumptions is asymptotically tight, up to $O(1/n)$ terms. We continue by illustrating the importance of stipulating that ranking technologies in our setting are stochastically consistent.
We show that when one or more ranking strategies can be arbitrary there can be pathological scenarios where even a dominant strategy equilibrium can be vastly worse than the optimal social welfare, by a factor linear in the number of firms. 
Finally, we conclude this section by revisiting the definition of stochastic consistency, introducing a relaxation of the definition that pertains to general advice spaces and quantifies the extent to which they violate stochastic consistency. We provide bounds that scale gracefully as stochastic consistency ``degrades''.

\subsection{A tight lower bound construction}
\begin{theorem}\label{thm:tight_ex}
    The Price of Anarchy bound of 2 is asymptotically tight.
\end{theorem}
\begin{proof}
    Consider a symmetric game with $n$ firms and two ranking technologies for each one: common ranking technologies $A$ and $A'$. 
    Further, consider the following permutation invariant distribution on the candidate values; there are $2^n$ candidates, a random pair of them have values $1$ and $v(n) \in (0, 1)$ (to be determined below), while all other candidates have value $0$. 
    
    The ranking technology $A$ produces rankings that always put the candidate with value $1$ on the first position; the rest of the ranking constitutes a uniform permutation on the remaining candidates. 
    On the other hand, $A'$ always puts the two non-zero value candidates in the first $n$ positions, uniformly at random and the rest of the ranking also constitutes a uniformly random permutation on the remaining candidates.

    Notice that this game instance satisfies the assumptions that lead to our bound; the ranking technologies are stochastically consistent. Also the candidate distribution is permutation invariant which implies incentive compatibility for the relaxed model.

    The optimum profile is for all firms to choose $A'$, resulting in the maximum welfare of $1+v(n)$. Observe that the utility of a firm choosing $A$, while everyone else also chooses $A$ is:
    \[u(A, A^{n-1}) = 1/n + v(n)\cdot \Theta(2^{-n})\]

    Now, consider a deviation of one such firm to strategy $A'$. With probability $1/n$ the firm chooses first and her conditional expected utility is $1/n(1 + v(n))$. 
    With the remaining $1-1/n$ probability the firm is not choosing first and cannot hire the candidate with value 1, but still hopes to get the candidate with value $v(n)$. 
    This either happens because the candidate with value $v(n)$ ranks above all candidates with value 0 (probability $\frac{1}{n-1}$) and none of the earlier firms using strategy $A$ picked that candidate (probability $1 - O(n \cdot 2^{-n})$) or because the candidate with value $v(n)$ is ranked below at least one candidate with value 0, but all higher-ranked candidates were chosen by earlier firms using strategy $A$ (probability $O(n \cdot 2^{-n})$). 
    In total, then, the expected utility of using strategy $A'$ against profile $A^{n-1}$ is
    \begin{align*}
        u(A', A^{n-1})  & = \tfrac1n \cdot \tfrac1n \left(1 + v(n) \right)
        \, + \, \left(1 - \tfrac1n\right)\!\Big[\tfrac{1}{n-1}(1-O(n2^{-n})) \, + \, O(n2^{-n})\Big] v(n) \\
        & \leq \tfrac{v(n)}{n} + \tfrac{1}{n^2} + v(n)(\tfrac{1}{n^2} + O(n2^{-n}))
    \end{align*}
    We can choose $v(n)$ so that $A^n$ is an equilibrium, to that end we want $u(A, A^{n-1})$ $> u(A', A^{n-1})$ to hold and our goal is to maximize $1 + v(n)$, the asymptotic ratio of $SW(A'^n)/SW(A^n)$.

    If we ignore the asymptotics, we ask to maximize $v(n)$ subject to the constraint $1/n > v(n)(1/n + 1/n^2) + 1/n^2$. 
    Thus, for some $v(n) = 1 - 2/(n+1) - o(1/n)$ we have that $A^n$ is indeed an equilibrium. We can estimate the social welfare of the strategy profiles $A^n$ and $A'^{\, n}$ as follows:
    \begin{align*}
        SW(A^n) & = 1 + n \cdot v(n) \cdot \Theta(2^{-n}) \\
        SW(A'^{\, n}) & = 1 + v(n) = 2 - 2/(n+1) - o(1/n) .
    \end{align*}
    From these estimates we conclude that the price of anarchy is $2 - O(1/n)$.
\end{proof}

\subsection{Arbitrary advice spaces lead to arbitrarily bad results}
The next example, while unrealistic, shows that in the setting of \cite{monoculture_kleinberg_raghavan} even a dominant strategy equilibrium can be arbitrarily bad, when ranking technologies are not stochastically consistent. 
\begin{proposition}
    The price of anarchy is $\Omega(n)$ (linear in the number of firms), for an arbitrary advice space, in the Obedience-Constrained Random Serial Dictatorship. 
    Further, there exist game instances where there is a (strict) dominant strategy, the socially optimal profile has no firm using the dominant strategy and the optimal welfare is a factor of $\Theta(n)$ larger than that of the dominant strategy equilibrium.
\end{proposition}
\begin{proof}
    We consider $n$ firms as usual and $2n$ candidates. 
    The candidate distribution is such that uniformly at random, a set of $n-1$ of them have value $0$, one of them has value $n$ and the $n$ rest have unique values from the increasing sequence $\{a_i\}_{i \in [n]}$ with $a_{j} = 1 + \frac{j(j-1)}{n-j+1} -\frac{(j-1)(j-2)}{n-j+2} - \eps/n^2$ for some small $\eps > 0$. Notice that this distribution is permutation invariant.
    
    Firms have access to two ranking technologies, a common $A$ and a private $H$ (of the same distribution for any firm). 
    Rankings from $A$ always rank the candidate with value $n$ at the first position, the $n-1$ candidates with value $0$ in the next positions and then everyone else arbitrarily. 
    On the other hand, rankings from $H$ rank the candidates with values from the sequence in the first $n$ positions in increasing order and everyone else arbitrarily in the next positions.

    For brevity, let $A^kH^{n-k}$ denote the profile where $k$ firms have chosen $A$ and $n-k$ chose their private $H$. 
    Clearly, $u(A, A^{n-k-1}H^{k}) = 1 + \frac{k}{n-k}$ by construction of $A$: a firm choosing $A$ gets non-zero value ($n$) if it is the first to make a selection among the $n-k$ firms choosing $A$. 
    Now, also observe that $u(H, A^{n-k-1}H^{k}) = \frac{1}{k+1}\sum_{i=1}^{k+1}a_i = 1 + \frac{k}{n-k} - \eps/n^2$. Thus, $A$ is indeed a dominant strategy.

    It's clear that the socially optimal profile is $H^n$: changing any firm to employ $A$ results in the candidate with value $n$ getting hired, but consequently, society loses out on the candidate with value $a_n = \Omega(n^2)$. Further, any profile with two or more firms employing $A$ is dominated by the profile where only one out of them employs $A$, and all others use $H$. The optimal social welfare then is $SW(H^n) = \frac{n(n+1)}{2} - \eps/n$, while the dominant strategy equilibrium welfare is just $n$.
\end{proof}

\section{On stochastic inconsistency}\label{sec:relaxation-of-sc}
We have argued in previous sections that stochastic consistency is a broadly encompassing condition. It not only provides guarantees that 
bound the social welfare loss due to monoculture in the setting of \cite{monoculture_kleinberg_raghavan}, but also 
extends these guarantees to an extremely wide range of ranking distributions (including the standard ones used in models of stochastic choice) and to the more general setting of advice spaces and the model detailed in \Cref{sec:3prelim}.
Further, its definition is not merely a technical measure of ``well-behavedness'', but corresponds to an important desideratum for ranking: the relative ranking of two candidates should not depend on that of other people, or on characteristics that do not relate to their ability.

In an ideal world people and any advice they use to make decisions would follow this principle. However, obviously, the world is far from ideal: implicit and systemic bias is prevalent. This becomes ever more relevant with the rise of machine learning models and LLMs and, crucially, their use for advice in decision making. The biases of large language models have been documented, for example, in \cite{explicit_bias,bias_iceberg} and, most relevantly to our application, in \cite{hiring_bias}).
The question then becomes: what happens when stochastic consistency is violated, but not arbitrarily (that is, inconsistency can be measured by some factor)? We give an answer to this question as well. We begin by defining exactly this notion of the measure of stochastic “inconsistency".

\begin{definition}[\ssc]\label{def:delta_sc}
    A ranking technology $R$ is \ssc if the following holds for all $i < j \in [m]$, $k, \ell \in [m]$. If $x(k) \ge x(\ell)$ then almost surely:
    \[\pr[r\sim \cal{F}_R(x)]{(r_i, r_j) = (k, \ell)}[r^{-\set{k, \ell}}]\\
    \geq\\
    (1-\delta)\pr[r\sim\cal{F}_R(x)]{(r_i, r_j) = (\ell, k)}[r^{-\set{k, \ell}}]\]
\end{definition}

Fortunately, this measure of inconsistency allows us to state the following result.

\begin{restatable}{theorem}{deltasc}
       The Price of Anarchy of the Obedience-Constrained Random Serial Dictatorship mechanism, in which firms obtain their rankings from a $(1-\delta)$-Stochastically Consistent advice space, where $\delta \in [0, 1)$, is bounded above by $1 + \frac{1}{(1 - \delta)^2}$.
\end{restatable}

We prove the theorem in the appendix using a similar technique as our main result (the ``key'' lemma here, analogous to \Cref{lemma:deviation}, is the main obstacle and needs more careful consideration).

%% file: 7conclusion.tex
\section{Conclusion}
\label{sec:conclusion}
    Our work offers a tight bound on the social welfare loss in a hiring market with advice (something that can arise due to algorithmic monoculture or general selfish behavior) and identifies conditions and assumptions one should impose for such settings to be well behaved. The phenomenon of algorithmic monoculture in markets that can be modeled by the setting we investigate seems to not pose a grave risk to the overall welfare. 
    Arguing about how detrimental monoculture is in regards to the price of anarchy or social welfare loss in the models of~\cite{improved_bayes_risk_jagadeesan,monoculture_matching_peng_garg}, as well as other models of algorithmic monoculture, remains an enticing question for future work.

%% file: 8apdx.tex
\appendix
\section{Missing proofs}\label{sec:apdxA}

\lemmaDev*
\begin{proof}\label{lemma:dev/proof}
Consider an ordering of firms $\beta$ and a realization $\mbf{x}$ of the values. 
We will show that this inequality is true for any choice of these. This is obviously true when $\beta(i) = 1$, i.e. when firm $i$ chooses first. 
Thus we assume $\beta(i) > 1$, so then $i$ in both profiles has some non-zero number of firms going before her.

Let's start by assuming $s^*_{i}$ is a private ranking.  
When deviating to choosing $s^*_{i}$ in profile $\mbf{s}$ the firm will get the same candidate as in profile $\mbf{s}^*$ or someone ranked higher (as per the conditioned event), so intuition says she should not do worse in expectation. 
However, this is not true in general if $i$ is up against firms that use arbitrary (not stochastically consistent) technologies, as the event that $i$ has a higher ranked candidate available can be non-trivially correlated with the value of that candidate, see \Cref{ex:deviation} for example.

Condition on the set of candidates that were hired by the firms going before $i$ in profile $\mbf{s}$, e.g. say $C^{(\beta)}_i(\mbf{s}) = C$ for some arbitrary $C$.  Consider the ranking sample that $i$ draws using $s_i^*$, let it be $r$. 
When deviating to this strategy in profile $\mbf{s}$, then, $i$ will select some candidate $r_1^{-C}$. 
Let $\cal{E}_C$ denote the event that the vector of candidate values is $\mbf{x}$, the ordering of firms is $\beta$, the set $C^{(\beta)}_i(\mbf{s})$ is equal to $C$, and $c_i^{(\beta)}(\mbf{s}^*) \not\in C$. 
The event $\cal{E}_C$ can be partitioned into mutually exclusive events $\cal{E}_{j,C}$, where $j$ ranges from 1 to $m - |C|$ and $\cal{E}_{j,C}$ represents the subset of sample points in $\cal{E}_C$ that satisfy $c_i^{(\beta)}(\mbf{s}^*) = r_j^{-C}$. 
Conditional on event $\cal{E}_{1,C}$ we have $c_i^{(\beta)}(s_i^*,\mbf{s}_{-i}) = c_i^{(\beta)}(\mbf{s}^*) = r_1^{-C}$, so $x_i^{(\beta)}(s_i^*, \mbf{s}_{-i}) - x_i^{(\beta)}(\mbf{s}^*) = 0$. 
Accordingly, for the remainder of the proof we focus on the case $j>1$.

Now, notice that conditioning on $C^{(\beta)}_i(\mbf{s}) = C$ does not imply anything about $i$'s ranking $r$, by independence of rankings. 
However, assuming $j>1$ we know that $r_1^{-C}$ was not available in profile $\mbf{s^*}$, which could correlate with its value. 
Consider then whoever got that candidate in profile $\mbf{s}^*$  and let their ranking sample be $f$. 
The idea is that since she is also using a stochastically consistent technology, she will be ``making the right choice'' most of the time. 

Now, we condition on an arbitrary pair of candidates $k, \ell$, so that either $(r_1^{-C}, r_j^{-C}) = (k, \ell)$ or $(r_1^{-C}, r_j^{-C}) = (\ell, k)$. 
If we can show that for any such pair with $x(k) \geq x(\ell)$ the first event is the most likely, then we are done, since:
\begin{gather*}
\ex{x_i^{(\beta)}(s_i^*, \mbf{s}_{-i}) - x_i^{(\beta)}(\mbf{s}^*) \;\bigg\vert\; \mbf{x}, \beta, c_i^{(\beta)}(\mbf{s}^*) 
\notin C_i^{(\beta)}(\mbf{s})}
=\\
\sfrac{1}{2}\cdot\mathbb{E}_{C \sim C_i^{(\beta)}(\mbf{s})}\sum_{j > 1}\sum_{\substack{k, \ell \in [m]\\k\neq\ell}}(x(k) - x(\ell))\pr{\cal{E}_{j,C}}\bigg\{\pr{(r_1^{-C}, r_j^{-C}) = (k, \ell) \mid \cal{E}_{j,C}} \\
- \pr{(r_1^{-C}, r_j^{-C}) = (\ell, k) \mid \cal{E}_{j,C}}\bigg\}
\end{gather*}

Fix (condition) the arbitrary positions that $k$, $\ell$ could be ranked in rankings $r$ and $f$, say $u$ and $u+v$ in $r$ and $z$, $z+w$ in $f$ (note that they both have to ``agree'' on their relative positions by the conditioned event) and condition on everyone else's ranking (including the other positions of $r$ and $f$) for simplicity of the argument.

To complete the proof we make use of the following inequality,
 for any $u < u + v, z < z + w \in [m]$ and for any $k \neq \ell \in [m]$ with $x(k) \geq x(\ell)$, 
 when $\cal{E}'$ is an arbitrary event (consistent with the event $\cal{E}_{j,C}$) 
 that may depend on the rankings of any other technology besides $r$ and $f$, and on all positions 
 of all candidates $c \notin \{k, \ell\}$ in rankings $r$ and $f$.
\begin{equation} \label{eq:eprime}
\pr{\bmat{r_u \\ r_{u+v}} = \bmat{k\\\ell} = \bmat{f_z\\f_{z+w}} \;\bigg\vert\; \cal{E}'} \geq \pr{\bmat{r_u \\ r_{u+v}} = \bmat{\ell\\k} = \bmat{f_z\\f_{z+w}} \;\bigg\vert\; \cal{E}'}
\end{equation}
The inequality is valid since distinct technologies are mutually independent and stochastically consistent. 
The independence implies that the left side factorizes as 
$\pr{\bmat{r_u \\ r_{u+v}} = \bmat{k\\\ell} \;\bigg\vert\; \cal{E}'} \cdot \pr{\bmat{f_z\\f_{z+w}} = \bmat{k \\\ell} \;\bigg\vert\; \cal{E}'}$ and that the right side factorizes similarly. 
Stochastic consistency then implies that each factor on the left side is greater than or equal to the corresponding factor on the right side.

\newcommand{\ekluvzw}{\ensuremath{\cal{E}^{k \ell}_{uvzw}}}
To see that inequality~\eqref{eq:eprime} suffices to complete the proof of the lemma, define $\ekluvzw$
to be the event that $\{k,\ell\} = \{f_z, f_{z+w}\} = \{r_u, r_{u+v}\} = \{r_1^{-C}, r_j^{-C}\}$
(as \emph{unordered} sets) and observe that $\ekluvzw$ can be decomposed as a disjoint union of events $\cal{E}'$ of the type defined above. Hence, by taking a weighted average of inequalities~\eqref{eq:eprime} for all the constituent sets $\cal{E'}$ that compose $\ekluvzw$, we obtain
\begin{equation*} 
\pr{\bmat{r_u \\ r_{u+v}} = \bmat{k\\\ell} = \bmat{f_z\\f_{z+w}} \;\bigg\vert\; \ekluvzw} \geq \pr{\bmat{r_u \\ r_{u+v}} = \bmat{\ell\\k} = \bmat{f_z\\f_{z+w}} \;\bigg\vert\; \ekluvzw} .
\end{equation*}
Now we calculate
\begin{gather*}
    \pr{(r_1^{-C}, r_j^{-C}) = (k, \ell) \mid \cal{E}_{j,C}} - \pr{(r_1^{-C}, r_j^{-C}) = (\ell, k) \mid \cal{E}_{j,C}}\\
    =\sum_{\substack{u, v, z, w \in [m]\\u+v>u\\w+z>w}}\pr{\ekluvzw \mid \cal{E}_{j,C}}\Bigg\{\pr{\bmat{r_u \\ r_{u+v}} = \bmat{k\\\ell} = \bmat{f_z\\f_{z+w}} \;\bigg\vert\; \ekluvzw}\\ \hspace{5cm} - \pr{\bmat{r_u \\ r_{u+v}} = \bmat{\ell\\k} = \bmat{f_z\\f_{z+w}} \;\bigg\vert\; \ekluvzw}\Bigg\}
\end{gather*}
and deduce that the sum is non-negative because each
summand is non-negative.

Finally, to conclude the proof we have to argue about ``common'' strategies, but observe that this proof can also be directly extended for that case too. 
There's only a few small changes, specifically in this case conditioning on $C^{(\beta)}_i(\mbf{s}) = C$ might affect \emph{some} of $i$'s ranking, but crucially, stochastic consistency implies that it won't affect that it's more likely a pair specific pair of candidates in $r^{-C}$ will be ranked correctly rather than not. 
Additionally, we won't be conditioning on everyone else's ranking, but only on those before $i$ in both profiles, except the positions as defined previously. 
This again might affect some of $i$'s ranking, but not the positions we care about (by the event), and so the arguments flow in the following fashion. 
If the firm who took candidate $r_1^{-C}$ (in profile $\mbf{s}$) is using another independent technology, the proof follows the same blueprint as previously. 
If the firm that took the candidate is using the same common technology, then the proof is simpler, as we just have to argue about the common technology's stochastic consistency.  
\end{proof}

\begin{proposition}\label{ex:deviation}
    \Cref{lemma:deviation} does not necessarily hold if firm $i$ is using a stochastically consistent technology, but not her competitors.
\end{proposition}
\begin{proof}
    Assume the following for simplicity of the argument: $m = n+3$, $\mbf{x}$ is such that $x(i), i \in [m]$ is a decreasing sequence and $\beta$ is such that $i$ is last to choose ($\beta(i) = n$). 
    Now, consider the following about the profiles $\mbf{s}_{-i}$ and $\mbf{s}^*_{-i}$: the former is such, that they always choose the best candidates, but never identify $1$ to be such, meaning they will always hire the candidates of $[m]\setminus\{1, m-1, m\}$, leaving the rest for $i$. 
    The latter, profile $\mbf{s}_{-i}^*$, is not very well informed: they always choose the \emph{worst} candidates, leaving $\{1, 2, 3\}$ for $i$ to choose. 
    Now, suppose $i$ is employing a stochastically consistent ranking technology in $\mbf{s}^*$, such as the Mallow's model with the Kendall-Tau distance. 
    Her utility in profile $\mbf{s}^*$ conditioned on the event of \Cref{lemma:deviation} is just $x(1)$. This is because the event is only true when she hires candidate $1$, as candidates $2$ and $3$ are never available in $\mbf{s}$. 
    What is her conditional utility in profile $\mbf{s}$ if she were to deviate to $s_i^*$? 
    It must be less than $x(1)$, because there is a non zero probability she ranked $m$ or $m-1$ above $1$ (even though in profile $\mbf{s}^*$ she would hire candidate $1$).
\end{proof}

\begin{theorem}\label{thm:smooth}
    The Obedience-Constrained mechanism with a stochastically consistent advice space is $(1,1)$-smooth.
\end{theorem}
\begin{proof}
    Consider any two profiles $\mbf{s}, \mbf{t}$. We have the following inequalities:
    
    \begin{align*}
        \sum_{i\in[n]}u_i(t_i, \mbf{s}_{-i}) &\ge \ex[\beta]{W^{(\beta)}_{\mathrm{available}}(t, s)} \tag{by \Cref{lemma:deviation}}\\
        &= SW(t) - \ex[\beta]{W^{(\beta)}_{\mathrm{snatched}}(t, s)}\\
        &\ge SW(t) - SW(s)
    \end{align*}

    The second line follows from the definition of $W^{(\beta)}_{\mathrm{snatched}}(t, s)$ and $W^{(\beta)}_{\mathrm{available}}(t, s)$, as: \[\ex[\beta]{W^{(\beta)}_{\mathrm{snatched}}(t, s)} + \ex[\beta]{W^{(\beta)}_{\mathrm{available}}(t, s)} = SW(t)\] Similarly, the third line follows from observing that $SW(s) \ge \ex[\beta]{W^{(\beta)}_{\mathrm{snatched}}(t, s)}$.
\end{proof}

\propdist*
\begin{proof}\label{prop:cvx_sum}
Let $d$ be the distance function in question. Without loss of generality, it suffices to prove that for any arbitrary $k < \ell$, any permutation $\pi$ with $i = \pi^{-1}(k) < j = \pi^{-1}(\ell)$ and the permutation $\pi'$ which arises from $\pi$ by swapping the position of $k$ and $\ell$, the following holds:
\[d(\pi, (1, \hdots, m)) \le d(\pi', (1, \hdots, m))\]

Observe that by the form of the distance function $d$ and how $\pi$ and $\pi'$ relate:
\begin{align*}
    d(\pi, (1, \hdots, m)) - d(\pi', (1, \hdots, m)) &= g(|\pi^{-1}(k) - k|) + g(|\pi^{-1}(\ell) - \ell|)\\
    &\;\;\;\;\;\;\;\;\;\;\;- \l(g(|\pi^{-1}(\ell) - k|) + g(|\pi^{-1}(k) - \ell|)\r)\\
    &= g(|i - k|) - g(|i - \ell|) -\l(g(|j - \ell| - g(|j - k|)\r)
\end{align*}

We define $f : x \mapsto g(|x - k|) - g(|x - \ell|)$, with our goal being to show that $f$ is non-decreasing. For ease of argument we assume $g$ is differentiable, but the same follows for a general non-differentiable $g$. Suppose $x \ge \ell > k$, we have then that $f(x) = g(x - k) - g(x - \ell)$ and so $f'(x) \ge 0$, as $g$ is convex. Similarly, for $x \le k$ we have $f(x) = g(k - x) - g(\ell - x)$ and so $f'(x) = g'(\ell - x) - g'(k - x) \ge 0$, again by convexity of $g$. Now, let $x \in (k, \ell)$, we have now $f(x) = g(x - k) - g(\ell - x)$ and $f'(x) = g'(x-k) + g'(\ell - x) \ge 0$, since $g$ is non-decreasing. Further, since $g$ is convex it must also be continuous (which implies continuity of $f$) and thus $f$ must be non-decreasing in its entire domain. 

The claim follows from observing that $d(\pi, (1, \hdots, m)) - d(\pi', (1, \hdots, m)) = f(i) - f(j)$ and since $i < j$ and $f$ is non-decreasing, it is $d(\pi, (1, \hdots, m)) - d(\pi', (1, \hdots, m)) \le 0$, as needed.
\end{proof}

\thmSchur*
\begin{proof}\label{thm/proof}
   Let $f: \R^m \to \R$ be the pdf of the additive noise vector. 
   We fix arbitrary $i < j$ and $k, \ell \in [m]$ and condition on the candidate value vector, where without loss of generality we assume $x(k) \ge x(\ell)$.
    
    Let $r$ be the random ranking of the candidates induced by the additive noise scoring procedure and let $\cal{E}_{ij}(k,\ell,\mbf{c})$ denote the event that $\{(r_i^{-C}, r_j^{-C}) = (k, \ell)\}$ and $r^{-\set{k, \ell}} = \mbf{c}$.
    Further, let $M_{ij}$ be the set of joint noise vectors $\bm{\eps}$ that realize event $\cal{E}_{ij}(k,\ell,\mbf{c})$. 
    Define $\cal{E}_{ij}(k,\ell,\mbf{c})$ and $M_{ji}$ accordingly. 
    If we show that for all $\mbf{c}$ the following difference is non-negative we are done:
    \begin{align*}
        \pr{\cal{E}_{ij}(k, \ell, \mbf{c})} - \pr{\cal{E}_{ji}(k, \ell, \mbf{c})} &= \int_{M_{ij}}f(\bm{\eps}) \, d\bm{\eps} - \int_{M_{ji}}f(\bm{\eps}) \, d\bm{\eps}
    \end{align*}
    
    Let $\delta_{k\ell} = x(k)- x(\ell) \geq 0$ and define the below bijection:
    
    \[\xi : \R^m \ni \bm{y} \mapsto \bm{y} + \delta_{k\ell}(\bm{e}_k - \bm{e}_\ell) \in \R^m\]

    Further let $B$ be the matrix that swaps coordinates $k, \ell$ and define the bijection $t = \xi\circ B$. 
    Observe that $t$ maps elements of $M_{ij}$ to $M_{ji}$ and vice-versa: rankings remain the same, except items $k$ and $\ell$ are swapped. 
    To see this, let $\bm{\eps} \in M_{ij}(\mbf{c})$, with $\eps(m)$ the additive noise to candidate $m$. 
    After the transformation the noise vector remains identical, except the $k$-th and $\ell$-th component that would now be $\eps(\ell) - \delta_{k\ell}$ and $\eps(k) + \delta_{k\ell}$, resulting in scores $x(\ell) + \eps(\ell)$ and $x(k) + \eps(k)$ for items $k$ and $\ell$ respectively (the scores of the rest of the candidates remain unchanged). 
    Thus, the only difference in this new ranking is that $k$ and $\ell$ are swapped. 
    Now, clearly $t$ is Lebesgue measure preserving and since it is a bijection between $M_{ij}$ and $M_{ji}$ it's the case that: $t^{-1}(M_{ji}) = M_{ij}$.

    We have then, that:
    \begin{align}\label{inter}
        \pr{\cal{E}_{ij}(k, \ell, \mbf{c})} - \pr{\cal{E}_{ji}(k, \ell, \mbf{c})} &= \int_{M_{ij}}f(\bm{\eps}) \, d\bm{\eps} - \int_{t^{-1}(M_{ji})}f(t(\bm{\eps})) \, d\bm{\eps}\nonumber\\
        &= \int_{M_{ij}}\big(f(\bm{\eps}) -f(t(\bm{\eps})) \big) \, d\bm{\eps}
    \end{align}

    Now, observe that if $\bm{z} = t(\bm{\eps})$, then $z(r) = \eps(r)$ for all $r \in [m]\setminus\{i, j\}$, but $z(k) = \eps(\ell) - \delta_{k\ell}$ and $z(\ell) = \eps(k) + \delta_{k\ell}$, as we noted before. 
    Recall then that, since $\bm{\eps} \in M_{ij}$ we must have that $k$ is ranked above $\ell$:
    \[x(k) + \eps(k) \geq x(\ell) + \eps(\ell)\] 

    Which implies:
    \[\eps(\ell) - \delta_{k\ell} \leq \min\{\eps(k), \eps(\ell)\}.\] 

    By the same observation we also conclude that:
    \[\eps(k) + \delta_{k\ell} \geq \max\{\eps(k), \eps(\ell)\}\]
    
    Thus, $t(\bm{\eps})$ majorizes $\bm{\eps}$ and by Schur-concavity of $f$ we have that the intergrands in \Cref{inter} are always non-negative, leading to $\pr{\cal{E}_{ij}(k, \ell, \mbf{c})} - \pr{\cal{E}_{ji}(k, \ell, \mbf{c})} \ge 0$, which implies what's desired.
\end{proof}

Here, we present an equivalent definition for \hyperref[def:sc]{Stochastic Consistency}, which will be useful in some of our proofs.
\begin{proposition}
        A ranking technology $R$ is \emph{stochastically consistent} iff for any two subsets $M$ and $C$ of $[m]$ with cardinality $q$ and any enumeration $\{m_i\}_{i \in [q]}$, $\{c_i\}_{i \in [q]}$ of them, for any realization $\mbf{x}$ of the joint value vector, the following holds for all $i, j \notin M$, $k, \ell \notin C$ with $(i-j)(x(k)-x(\ell)) \leq 0$:
    \begin{gather*}\pr[r\sim \cal{F}_R(x)]{(r_i, r_j) = (k, \ell) \mid (r_{m_1}, \hdots, r_{m_q}) = (c_1, \hdots, c_{q})}\\
    \geq\\
    \pr[r\sim\cal{F}_R(x)]{(r_i, r_j) = (\ell, k)}[(r_{m_1}, \hdots, r_{m_q}) = (c_1, \hdots, c_q)],
    \end{gather*}
    (whenever the conditioning event has non-zero probability).
\end{proposition}
\begin{proof}
It's clear that this definition implies the original. The same is true for the inverse implication: we get the inequality for any relevant subset and their enumeration, using the law of total probability.
\end{proof}

\vspace{0.5\baselineskip}\propIC*\vspace{-0.24cm}
\begin{proof}\label{prop/proof}
    We will show that no firm has any incentive to strategize and ask for a candidate other than the highest ranked that is available from her sample, even if she knew not just the candidates hired before her, but also the technologies used. 
    Suppose then she has a strategy to permute her sampled ranking, based on what she witnessed (technologies used and candidates hired) and her ranking sample.

    Let $\mbf{s} := (s_1, \hdots, s_k)$ for some $k\geq0$ denote the arbitrary strategies of the firms going before her in order and let $D(\mbf{s})$ denote the (vector of) candidates that each hired. 
    First, consider the event where the firms hire some arbitrary candidates in that order, i.e. say $D(\mbf{s}) = (c_1, \hdots, c_k)$. 
    We claim that the joint distribution of the value of candidates $C' = [m] \setminus \{c_1,\hdots,c_k\}$ conditioned on $D(\mbf{s}) = (c_1, \hdots, c_k)$ remains permutation invariant. This is because the event $D(\mbf{s}) = (c_1, \hdots, c_k)$ is (trivially) \emph{symmetric} with respect to these candidates and so permutation invariance of the original joint value vector, shall imply permutation invariance for the joint value vector of the candidates in $C'$ too, conditioned on this event.

    This shows that when a firm witnesses the event $D(\mbf{s}) = (c_1, \hdots, c_k)$ she doesn't have any incentive to strategize with respect to the identities of the candidates she sees in the partial ranking $r^{-D}$. 
    Thus, her strategy can as well depend only on the event $D(\mbf{s}) = (c_1, \hdots, c_k)$ by possibly choosing some later \emph{position} in $r^{-D}$ than the first (and shall not depend on the sample $r$ she sees). 
    Say then, whenever she witnesses $D(\mbf{s}) = (c_1, \hdots, c_k)$ she chooses the $q$-th ranked candidate from her list of available candidates for some $q > 1$, ie she seeks to hire $r_q^{-D}$ instead of $r_1^{-D}$. 
    The event $D(\mbf{s}) = (c_1, \hdots, c_k)$ implies that for some $m \geq 0$ and some suitable indices $\ell_i$ we have that $(r_1, \hdots, r_m) = (c_{\ell_1}, \hdots, c_{\ell_m})$ ($m > 0$ implies that this ranking is sampled from an algorithmic vendor and $\mbf{s}$ has some firms employing the same advice). But now notice the following, because of the independence of distinct technologies:
    \[\ex{x(r_1^{-D}) - x(r_q^{-D}) \mid D(\mbf{s}) = D} = \ex{x(r_1^{-D}) - x(r_q^{-D}) \mid (r_1, \hdots, r_m) = (c_{\ell_1}, \hdots, c_{\ell_m})}\]
    And now it's not hard to see that stochastic consistency ensures that:
    \[\ex{x(r_1^{-D}) - x(r_q^{-D}) \mid (r_1, \hdots, r_m) = (c_{\ell_1}, \hdots, c_{\ell_m})} \geq 0\]

    This is because:
    \begin{gather*}
        \ex{x(r_1^{-D}) - x(r_q^{-D}) \mid (r_1, \hdots, r_m) = (c_{\ell_1}, \hdots, c_{\ell_m})}\\
        = \frac{1}{2}\cdot\mathbb{E}\sum_{\substack{k, \ell \in C'\\ k \neq \ell}}\sum_{\substack{i, j \in[m]\\ i < j}}(x(k) - x(\ell))\cdot\Big\{\pr{(r_i, r_j)=(k, \ell) \mid \cal{E}_D \land \cal{E}} \\
        \hspace{5cm}- \pr{(r_i, r_j)=(\ell, k) \mid \cal{E}_D \land \cal{E}}\Big\}
    \end{gather*}
    Where $\cal{E}_D$ represents the (random) rankings of the candidates in $D$, that are consistent with $\{r_1^{-D}, r_q^{-D}\} = \{k, \ell\} = \{r_i, r_j\}$ and $\cal{E}$, which is the event $(r_1, \hdots, r_m) = (c_{\ell_1}, \hdots, c_{\ell_m})$. 
    Since the statement of stochastic consistency implies that for any arbitrary conditioning on the ranking of a set of candidates (disjoint from those of interest), the above quantity must be non-negative.\\

    This concludes the proof (as the conditioning $D(S) = D$ was arbitrary).
\end{proof}

\begin{lemma}\label{lemma:delta_dev}
Consider an Obedience-Constrained Random Serial Dictatorship mechanism with a \ssc advice space, for $\delta \in [0,1)$. Let $c_i^{(\beta)}(\mbf{s})$ denote the candidate hired  by firm $i$ in profile $\mbf{s}$ when the order of firms is $\beta$. Let also $x_i^{(\beta)}(\mbf{s}) := x(c_i^{(\beta)}(\mbf{s}))$ be the value of this candidate. Further, define $C_i^{(\beta)}(\mbf{s})$ to be the set of candidates hired by the firms going before $i$ in profile $s$.\\

For any two pure strategy profiles $s$, $s^*$ the below holds:
\[\ex{x_i^{(\beta)}(s_i^*, \mbf{s}_{-i}) - (1-\delta)^2x_i^{(\beta)}(\mbf{s}^*) \;\bigg\vert\; \mbf{x}, \beta, c_i^{(\beta)}(\mbf{s}^*) \notin C_i^{(\beta)}(\mbf{s})} \geq 0\]
where the expectation is over the randomness of the samples drawn from the ranking technologies and $\mbf{x}$ represents the joint value vector of the candidates.
\end{lemma}

\begin{proof}
    Let $\cal{E}$ be the event that the joint value vector is $\mbf{x}$, the order is $\beta$ and $c_i^{(\beta)}(\mbf{s}^*) \notin C_i^{(\beta)}(\mbf{s})$. 
    
    Now, similar to the \hyperref[lemma:dev/proof]{proof} of Lemma 4.3, 
    define the event $\cal{E}_{j,C}$ (where $j$ ranges from 1 to $m - |C|$) to be such that $C^{(\beta)}_i(\mbf{s}) = C$ and $c_i^{(\beta)}(\mbf{s}^*) = r_j^{-C}$.  Observe that we can write:
    \begin{gather}\label{dev_init}
        \ex{x_i^{(\beta)}(s_i^*, \mbf{s}_{-i}) - x_i^{(\beta)}(\mbf{s}^*) \;\bigg\vert\; \cal{E}}\nonumber\\
        =\mathbb{E}_{C \sim C_i^{(\beta)}(\mbf{s})}\sum_{j > 1}\sum_{\substack{k, \ell \in [m]\\x(k) > x(\ell)}}(x(k) - x(\ell))\pr{\cal{E}_{j,C}}\bigg\{\pr{(r_1^{-C}, r_j^{-C}) = (k, \ell) \mid \cal{E}_{j,C}}\nonumber \\
        \hspace{6cm}- \pr{(r_1^{-C}, r_j^{-C}) = (\ell, k) \mid \cal{E}_{j,C}}\bigg\}
    \end{gather}
    It's not hard to see that for a $(1 - \delta)$-Stochastically Consistent advice space we have for any two sample rankings $r, f$ (possibly the same, as in the \hyperref[lemma:dev/proof]{proof} of Lemma 4.3) 
    and valid indices:
    \begin{equation} \label{eq:eprime_delta}
        \pr{\bmat{r_u \\ r_{u+v}} = \bmat{k\\\ell} = \bmat{f_z\\f_{z+w}} \;\bigg\vert\; \cal{E}'} \geq \l(1 - \delta\r)^2\pr{\bmat{r_u \\ r_{u+v}} = \bmat{\ell\\k} = \bmat{f_z\\f_{z+w}} \;\bigg\vert\; \cal{E}'}
    \end{equation}

    But then we can see from (3) that (define $P^{j, C}_{k, \ell} := \pr{\cal{E}_{j,C}}\pr{(r_1^{-C}, r_j^{-C}) = (k, \ell) \mid \cal{E}_{j,C}}$ for brevity):
    \begin{align}
        &\ex{x_i^{(\beta)}(s_i^*, \mbf{s}_{-i}) - x_i^{(\beta)}(\mbf{s}^*) \;\bigg\vert\; \cal{E}}\nonumber\\
        &\hspace{1cm}\geq ((1-\delta)^2 - 1)\cdot\mathbb{E}_{C \sim C_i^{(\beta)}(\mbf{s})}\sum_{j > 1}\sum_{\substack{k, \ell \in [m]\\x(k) > x(\ell)}}(x(k) - x(\ell))P^{j, C}_{\ell, k}\\
        &\hspace{1cm}=((1-\delta)^2 - 1)\cdot\mathbb{E}_{C \sim C_i^{(\beta)}(\mbf{s})}\sum_{j > 1}\sum_{\substack{k, \ell \in [m]\\x(k) > x(\ell)}}\l\{x(k)P^{j, C}_{\ell, k} + x(\ell)P^{j, C}_{k, \ell}\r\}\nonumber\\
        &\hspace{1cm}\;\;\;\;+(1 - (1-\delta)^2)\cdot\mathbb{E}_{C \sim C_i^{(\beta)}(\mbf{s})}\sum_{j > 1}\sum_{\substack{k, \ell \in [m]\\x(k) > x(\ell)}}x(\ell)\l(P^{j, C}_{k, \ell} +P^{j, C}_{\ell, k}\r)\\
        &\hspace{1cm}\ge ((1-\delta)^2 - 1)\cdot\ex{x_i^{(\beta)}(\mbf{s}^*)\cdot\ind{x_i^{(\beta)}(s_i^*, \mbf{s}_{-i}) \neq x_i^{(\beta)}(\mbf{s}^*)} \;\bigg\vert\; \cal{E}}\\
        &\hspace{1cm}\ge ((1-\delta)^2 - 1)\ex{x_i^{(\beta)}(\mbf{s}^*) \;\bigg\vert\; \cal{E}}
    \end{align}

    To see line (7), note first that $1 - (1-\delta)^2 \ge 0$ for $\delta \in [0,1]$ and thus, the second term is non-negative. Then, observe that the term we lower bound in line (7) exactly represents the first term in the preceding line: it “counts" the expected value of the candidate hired by firm $i$ in profile $\mbf{s}$, only when this candidate is different than the one hired in $\mbf{s}$ (if she were to deviate to $s^*_i$). The final inequality follows from the fact that the factor is non-positive and candidate values are non-negative.
    
\end{proof}
\deltasc
\begin{proof}
    The proof is similar to that for Stochastically Consistent advice spaces. Using the same notation we have two cases.

    \paragraph{Case 1:} When \underline{$\ex[\beta]{\snatched} \geq \l(1 - \frac{1}{1+(1-\delta)^2}\r)\cdot SW(\mbf{s}^*)$}:\\

    Trivially then, we have that $\frac{SW(\mbf{s^*})}{SW(\mbf{s})} \le \l(1 - \frac{1}{1+(1-\delta)^2}\r)^{-1} = 1 + \frac{1}{(1 - \delta)^2}$.

    \paragraph{Case 2:} When \underline{$\ex[\beta]{\snatched} < \l(1 - \frac{1}{1+(1-\delta)^2}\r)\cdot SW(\mbf{s}^*)$}:\\

    By \Cref{lemma:delta_dev} we have that:
    \[SW(\mbf{s}) \ge (1-\delta)^2\cdot \ex[\beta]{\avail}\]

    By the case assumption we know $\ex[\beta]{\avail} \ge \frac{1}{1+(1-\delta)^2}SW(\mbf{s}^*)$ and so since $\mbf{s}$ is an equilibrium:
    \[SW(\mbf{s}) \ge (1-\delta)^2\frac{1}{1+(1-\delta)^2}\cdot SW(\mbf{s}^*)\]

    Which gives \[\frac{SW(\mbf{s}^*)}{SW(\mbf{s})} \le 1 + \frac{1}{(1 - \delta)^2}\].
\end{proof}